\documentclass[letterpaper,twocolumn,10pt]{article}
\usepackage{usenix}

\usepackage{graphicx}
\usepackage{amsfonts} 
\usepackage{mathtools} 
\usepackage{textcomp} 
\usepackage{amsthm}
\usepackage{amssymb}
\usepackage{thmtools}
\usepackage{subcaption}
\usepackage{amsmath}
\usepackage{pifont}
\newtheorem{theorem}{Theorem}
\newtheorem{lemma}[theorem]{Lemma} 

\newcommand{\id}[1]{\mathsf{#1}}
\newcommand{\domsalt}{\id{AUTIG/salt}}
\newcommand{\domfrag}{\id{AUTIG/frag}}
\newcommand{\salt}{\id{salt}}
\newcommand{\leaderpk}{\id{leader\_pk}}
\newcommand{\hdr}{\id{hdr}}
\newcommand{\Enc}{\id{enc}}
\newcommand{\Commit}{\id{Commit}}

\newcommand{\round}{\mathrm{round}}
\newcommand{\collect}{\mathrm{collect}}
\newcommand{\update}{\mathrm{update}}
\newcommand{\extract}{\mathrm{extract}}
\newcommand{\broadcast}{\mathrm{broadcast}}
\newcommand{\final}{\mathrm{final}}
\newcommand{\detect}{\mathrm{detect}}
\newcommand{\handoff}{\mathrm{handoff}}
\newcommand{\recovery}{\mathrm{recovery}}
\newcommand{\edge}{\mathrm{edge}}

\newcommand{\defer}{\mathrm{defer}}
\newcommand{\bad}{\mathrm{bad}}

\newcommand{\prev}{\mathrm{prev}}   
\newcommand{\batch}{\mathrm{batch}} 

\newcommand{\state}[1]{\ensuremath{\mathsf{#1}}}
\newcommand{\Solid}{\state{Solid}}
\newcommand{\Shaded}{\state{Shaded}}
\newcommand{\Blank}{\state{Blank}}

\usepackage[ruled,vlined,linesnumbered,noend]{algorithm2e}

\SetNlSkip{0.5em}                 

\RestyleAlgo{boxruled}           
\SetAlgoInsideSkip{smallskip}    
\SetAlFnt{\small}                
\SetAlCapFnt{\small}             
\SetAlCapNameFnt{\small}
\SetAlgoNlRelativeSize{-1}       
\SetKwInput{Input}{Input}
\SetKwInput{Output}{Output}
\SetKwComment{tcp}{$\triangleright$ }{}
\DontPrintSemicolon
\SetKwProg{Fn}{Procedure}{}{}
\usepackage{tikz}
\usepackage{amsmath}

\usepackage{filecontents}

\begin{document}

\date{}

\title{\Large \bf Proof-Carrying Fair Ordering: Asymmetric Verification for BFT via Incremental Graphs}

\author{
{\rm Pengkun Ren, Hai Dong, and Zahir Tari}\\
School of Computing Technologies, Centre of Cyber Security Research and Innovation,\\
RMIT University, Melbourne, Australia
\and
{\rm Nasrin Sohrabi}\\
Deakin University, Australia
\and
{\rm Pengcheng Zhang}\\
College of Computer Science and Software Engineering, Hohai University, China
}

\maketitle

\begin{abstract}
Byzantine Fault-Tolerant (BFT) consensus protocols ensure agreement on transaction ordering despite malicious actors, but unconstrained ordering power enables sophisticated value extraction attacks like front-running and sandwich attacks—a critical threat to blockchain systems. Order-fair consensus curbs adversarial value extraction by constraining how leaders may order transactions.  While state-of-the-art protocols such as Themis attain strong guarantees through graph-based ordering, they ask every replica to re-run the leader’s expensive ordering computation for validation—an inherently symmetric and redundant paradigm.
We present AUTIG, a high-performance, pluggable order-fairness service that breaks this symmetry. Our key insight is that verifying a fair order does not require re-computing it. Instead, verification can be reduced to a stateless audit of succinct, verifiable assertions about the ordering graph's properties. AUTIG realizes this via an asymmetric architecture: the leader maintains a persistent Unconfirmed-Transaction Incremental Graph (UTIG) to amortize graph construction across rounds and emits a structured proof of fairness with each proposal; followers validate the proof without maintaining historical state.
AUTIG introduces three critical innovations: (i) incremental graph maintenance driven by threshold-crossing events and state changes; (ii) a decoupled pipeline that overlaps leader-side collection/update/extraction with follower-side stateless verification; and (iii) a proof design covering all internal pairs in the finalized prefix plus a frontier completeness check to rule out hidden external dependencies. We implement AUTIG and evaluate it against symmetric graph-based baselines under partial synchrony. Experiments show higher throughput and lower end-to-end latency while preserving $\gamma$-batch-order-fairness.
\end{abstract}

\begin{table*}[t]
  \centering
  \caption{Architectural and complexity comparison of order-fair BFT protocols. $|B|$: batch size; $|\Delta B|$: incremental additions/updates; $|P|$: proof size.}
  \label{tab:comparison}
  \small
  \begin{tabular}{l|l|l|l|l}
    \hline
    \textbf{Property} & \textbf{Aequitas}~\cite{kelkar2020order} & \textbf{Themis}~\cite{kelkar2023themis} & \textbf{SpeedyFair}~\cite{mu2024separation} & \textbf{AUTIG (ours)} \\
    \hline\hline
    \multicolumn{5}{c}{\textit{Architectural Design}} \\
    \hline
    Liveness Guarantee & Weak & Standard (BFT) & Standard (BFT) & Standard (BFT) \\
    Ordering--Consensus Coupling & Tightly Coupled & Tightly Coupled & Decoupled (via OFO) & Decoupled (via Pipeline) \\
    Graph Data Model & Per-Round Transient & Per-Round Transient & Per-Round Transient & Persistent Cross-Round Graph \\
    Graph Update Model & Full Rebuild & Full Rebuild & Full Rebuild & Amortized Incremental Updates \\
    Verification Model & Symmetric & Symmetric & Symmetric & Asymmetric Proof-based Audit \\
    State Management & Stateless & Stateless & Mostly Stateless & Stateful at Leader (UTIG) \\
    \hline
    \multicolumn{5}{c}{\textit{Computational Complexity (per round)}} \\
    \hline
    Leader Ordering Work & $\mathcal{O}(|B|^2)$ & $\mathcal{O}(|B|^2)$ & $\mathcal{O}(|B|^2)$ & $\mathcal{O}(|\Delta B|^2)$ (amortized) \\
    Follower Verification Work & $\mathcal{O}(|B|^2)$ & $\mathcal{O}(|B|^2)$ & $\mathcal{O}(|B|^2)$ & $\mathcal{O}(|P|)$ (Proof Size) \\
    \hline
  \end{tabular}
\end{table*}
\section{Introduction}

Byzantine Fault-Tolerant consensus protocols form the bedrock of modern distributed systems, guaranteeing that a set of replicas can agree on a consistent, totally-ordered log of transactions despite the presence of malicious actors \cite{zhang2024reaching}. While these protocols robustly solve the problem of what to agree on, a critical, orthogonal dimension has recently garnered significant attention: the fairness of the agreed-upon order. In many high-stakes applications, particularly within the burgeoning ecosystem of Decentralized Finance (DeFi), the community \cite{daian2020flash, eskandari2019sok, qin2022quantifying, zhou2021high} has recognized that how the order is produced is as critical as the content itself. Unconstrained ordering power in the hands of a single protocol leader creates a fertile ground for sophisticated exploits, such as front-running \cite{baum2022sok} and sandwich attacks \cite{heimbach2022eliminating}. This value extraction, broadly termed Maximal Extractable Value (MEV) \cite{daian2020flash}, poses a significant threat to market integrity and user trust in both permissionless and permissioned \cite{castro1999practical, yin2019hotstuff} systems alike. To address this, the principle of order-fairness \cite{cachin2022quick, kelkar2022order, zhang2020byzantine, kursawe2020wendy, kursawe2020wendy} has emerged as a crucial security property for BFT consensus. The core idea is to democratize the ordering process: the final global transaction sequence should not be dictated by a single entity's whim but must instead reflect a supermajority consensus on the relative ordering of transactions as they are observed across the network.

\begin{figure}[t]
  \centering
  \begin{subfigure}{0.8\linewidth}
    \centering
    \includegraphics[width=\linewidth]{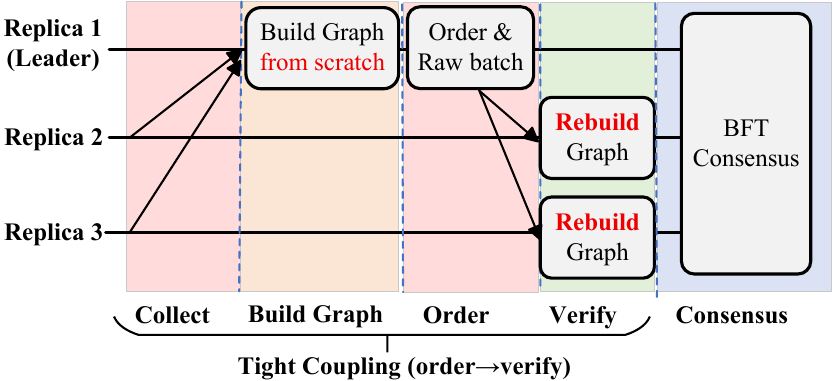}
    \caption{Themis rebuilds the dependency graph from scratch and followers re-execute symmetrically.}
    \label{fig:themis}
  \end{subfigure}

  \vspace{0.4em} 

  \begin{subfigure}{0.8\linewidth}
    \centering
    \includegraphics[width=\linewidth]{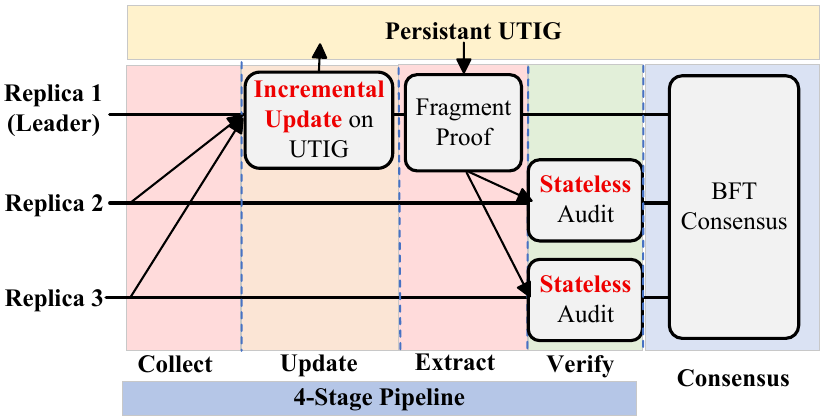}
    \caption{AUTIG: leader maintains a four-stage pipeline and a persistent UTIG with incremental updates and compact proofs enabling stateless verification.}
    \label{fig:autig}
  \end{subfigure}

  \caption{Execution-flow comparison of fair ordering engines.}
  \label{fig:swimlane-contrast}
\end{figure}

Significant strides have been made in this domain. A series of foundational, graph-based protocols~\cite{kelkar2020order, kelkar2023themis} have established a robust theoretical framework for achieving strong, relative order-fairness. These systems work by having a leader collect local orderings from replicas, construct a dependency graph representing collective preferences, and then algorithmically derive a verifiably fair proposal. However, a fundamental performance bottleneck emerges from the symmetric verification model inherent to this design line. To guard against a malicious leader, state-of-the-art protocols require every replica to independently re-execute the leader's entire, computationally intensive ordering process.  This requirement for system-wide, symmetric, and redundant re-computation fundamentally constrains throughput and finality. This core performance limitation stems from the symmetric verification model itself. While Themis~\cite{kelkar2023themis} exhibits this bottleneck within its tightly coupled, monolithic architecture, even decoupled approaches like SpeedyFair~\cite{mu2024separation} inherit it, as they merely relocate the expensive, from-scratch graph reconstruction to a parallel pipeline rather than fundamentally redesigning it.

In this paper, we explore a different architectural approach. Our central insight is that fairness verification can be decomposed from a holistic re-computation into an efficient audit of specific, locally-verifiable properties of the ordering graph. The integrity of a fair order rests on auditable predicates: (1) the aggregated weights of dependency edges correctly reflect the collected local orders, (2) the states of transactions are correctly derived, and (3) the final proposed sequence is a valid topological traversal of the resulting dependency graph. A leader can assert these properties and provide succinct evidence, transforming the followers' role from active re-constructors to lightweight auditors.

We introduce AUTIG, a novel order-fairness module that can be integrated with existing BFT systems and is built upon a novel asymmetric architecture.

\textbf{Stateful, Incremental Ordering with a Pipelined Engine.}  To eliminate the leader's primary overhead of per-round graph reconstruction, the AUTIG leader maintains a single, persistent data structure—the Unconfirmed Transaction Incremental Graph. This graph contains only transactions whose global order is still undecided, acting as a dynamic state machine for fairness computation. Instead of rebuilding the entire graph from all active transactions each round, the leader performs efficient, fine-grained incremental updates based on new local orders. This amortizes the cost of graph construction over many rounds. To further enhance performance, AUTIG employs a decoupled, four-stage pipeline for ordering: (1) Local Order Collection, (2) Incremental Graph Refresh, (3) Fair Prefix Extraction, and (4) Proposal Broadcasting. These stages operate concurrently, allowing the system to continuously process incoming data while the extraction logic periodically identifies and finalizes a prefix of the fair order.

\textbf{Lightweight, Proof-based Verification.} To eliminate the verification bottleneck for followers, AUTIG introduces a "propose-and-prove" model. The leader's proposal is accompanied by a compact, self-contained proof-of-fairness. This proof is not a generic cryptographic object but a structured set of evidence containing: (1) the minimal subgraph of newly finalized transactions and their internal dependencies, (2) attestations for the weights of these dependency edges, sufficient for followers to verify their directionality, and (3) assertions of the node states. This proof enables any replica to verify the proposal's fairness by only examining this small, self-contained data, without needing to maintain any historical graph state. Verification is thus reduced from \(O(|B|^{2})\) re-execution to \(O(|P|)\) proof auditing, where \(|P|\) is the size of the attached proof (internal-pair checks plus a frontier completeness set), i.e., near-linear in the evidence rather than quadratic in the batch. We also design a secure leader handoff protocol that preserves fairness guarantees without requiring the direct, costly transfer of the full UTIG.

This asymmetric design—\textbf{stateful leader, stateless followers}—fund-amentally diverges from prior symmetric models. Themis\cite{kelkar2023themis} requires all replicas to symmetrically rebuild the dependency graph and re-execute the ordering logic each round to verify a proposal. SpeedyFair\cite{mu2024separation}, while decoupling the process, retains this symmetric re-execution model within its parallel pipeline. By contrast, AUTIG's asymmetric state management allows it to uphold the strong fairness guarantees while achieving higher performance by eliminating the need for system-wide, per-round re-computation. The contrasting execution flows are illustrated in Fig.~\ref{fig:swimlane-contrast}. Throughout, we adopt the same \(\gamma\)-batch-order-fairness formulation as Themis~\cite{kelkar2023themis}, with $\tfrac{1}{2}<\gamma\le 1$.

Our contributions are as follows:
We design AUTIG, a pluggable order-fairness service whose leader maintains a persistent Unconfirmed Transaction Incremental Graph and updates it incrementally based on threshold-crossing events and state changes.

We develop a four-stage pipeline (collect \(\rightarrow\) update \(\rightarrow\) extract \(\rightarrow\) broadcast) that overlaps leader-side ordering with follower-side stateless verification, hiding ordering latency from the consensus critical path.

We introduce a propose-and-prove mechanism where compact proofs contain (i) all internal pair checks for the finalized prefix and (ii) a frontier completeness check preventing hidden external dependencies, reducing follower work from \(O(|B|^{2})\) to \(O(|P|)\).

We implement AUTIG and conduct extensive experiments, demonstrating that it significantly outperforms state-of-the-art symmetric protocols in both throughput and latency, making fast, strong order-fairness a practical reality for BFT systems.

\section{Background}

\subsection{BFT Consensus and the Ordering Problem}
Byzantine Fault-Tolerant consensus protocols are fundamental to the correctness of distributed systems operating in adversarial environments. They enable a network of $n$ replicas to maintain a consistent, replicated state machine despite up to $f$ of them exhibiting arbitrary, malicious behavior. Modern leader-based protocols, such as HotStuff~\cite{yin2019hotstuff}, have become a standard due to their communication efficiency and performance characteristics under partial synchrony. These protocols provide two core guarantees that are essential for distributed applications: Safety, ensuring all correct replicas agree on the same, totally-ordered log of transactions, and Liveness, ensuring all valid transactions submitted by clients are eventually included in the log.

While these guarantees are powerful, they are narrowly focused on the what (the content of the log) and not the how (the ordering within the log). In a typical leader-based protocol, a designated leader replica is responsible for proposing a new block of transactions in each round. The leader collects transactions from clients, sequences them into a list, and proposes this list for agreement. Other replicas vote on the proposed block, but their validation checks are typically limited to protocol rules like signature correctness or transaction validity, not the fairness of the internal ordering. This grants the leader complete autonomy over the selection and sequencing of transactions. This unilateral ordering power, while not violating the core safety or liveness properties of BFT, creates a critical and exploitable vulnerability: transaction order manipulation. A self-interested or malicious leader can leverage this power to reorder, delay, or strategically insert its own transactions to extract value. This phenomenon poses a tangible threat to the integrity of financial systems built on blockchains \cite{daian2020flash, malkhi2022maximal}.

\subsection{Order-Fairness and the Challenge}
\label{sec:background_order_fairness}

To mitigate leader-driven order manipulation, recent work has proposed order-fairness as a new safety property for BFT consensus. The central idea is to shift ordering authority from a single leader to the collective observation of the network. Two design lines have emerged. The first, using synchronized clocks as in Wendy and Pompe~\cite{zhang2020byzantine, kursawe2020wendy}, is fragile in partially synchronous settings due to clock skew and adversarial timestamping. The second, now dominant, line pursues relative order-fairness, using only the relative arrival order of transactions at replicas. Its ideal form, receive-order-fairness, requires that if a $\gamma$-fraction of correct replicas receive $tx_1$ before $tx_2$, then $tx_1$ must precede $tx_2$ globally. We follow the Themis-style \(\gamma\)-parameterization with \(\tfrac{1}{2}<\gamma\le 1\), where \(\gamma\) denotes the fraction of \emph{all} replicas (including Byzantine) whose observed precedence must constrain delivery. We formalize the “no-later-than” semantics using batch indices—i.e., \(b(tx_1)\le b(tx_2)\) whenever at least a \(\gamma\)-fraction report \(tx_1\prec tx_2\); the precise definition and feasibility bound are given in Sec.~\ref{sec:model}.

However, strict receive-order-fairness is unachievable due to the Condorcet paradox~\cite{gehrlein1983condorcet}. This classic problem from social choice theory shows that aggregating individual transitive preferences (e.g., each replica's consistent local timeline) can yield a cyclic majority relation. For instance, a majority might observe $tx_1 \prec tx_2$, another $tx_2 \prec tx_3$, and yet another $tx_3 \prec tx_1$, creating an unresolvable deadlock. Practical protocols therefore adopt batch-order-fairness~\cite{kelkar2020order,kelkar2023themis}: if a $\gamma$-fraction observe $tx_1$ before $tx_2$, then $tx_1$ is delivered no later than $tx_2$. This crucial relaxation allows transactions in a Condorcet cycle to be delivered together in a "batch," ensuring progress. Yet, it introduces a new liveness challenge: these cycles can become "chained" across rounds, potentially delaying finalization indefinitely (weak liveness). Systems like Themis~\cite{kelkar2023themis} address this via deferred ordering. Our work follows this relative, clock-free line, targeting the performance bottlenecks of its state-of-the-art implementations.

\subsection{Evolution of Graph-Based Architectures}
\label{sec:background_architectures}
The practical realization of batch-order-fairness has been shaped by a series of foundational, graph-based protocols. This evolution reveals a clear trajectory: from initial theoretical models that solved the impossibility of strict ordering but suffered from liveness flaws, to robust monolithic systems that introduced performance bottlenecks, and finally to decoupled architectures that improved latency but retained core computational inefficiencies. We analyze this progression in detail below and summarize the key distinctions in Table~\ref{tab:comparison}.

\textbf{Aequitas: The Foundational Graph Model and Its Liveness Flaw.} The initial breakthrough in applying batch-order-fairness came from \textbf{Aequitas}~\cite{kelkar2020order}. Aequitas was the first to formalize a graph-based approach to circumvent the Condorcet Paradox. Its core mechanism involves the leader collecting local orderings from replicas and constructing a dependency graph. In this graph, a directed edge $(tx_u, tx_v)$ exists if a supermajority of replicas observe $tx_u$ before $tx_v$. Transactions involved in a cycle (a Strongly Connected Component) are correctly identified and delivered together as an unordered batch, thus satisfying the batch-order-fairness definition.

While this elegantly solved the impossibility problem of receive-order-fairness, it introduced a critical weak liveness issue, a limitation later analyzed in detail by Themis~\cite{kelkar2023themis}. In the Aequitas model, Condorcet cycles can be "chained" together across multiple rounds. A transaction entering the system early could become part of an SCC, which in a subsequent round might merge with another SCC formed by newly arriving transactions. This process can repeat, creating arbitrarily long dependency chains. Consequently, a transaction could be indefinitely delayed from finalization, violating the standard liveness guarantees required by BFT systems.

\textbf{Themis: Achieving Liveness in a Monolithic, Symmetric Architecture.} Building directly on Aequitas's foundation, Themis~\cite{kelkar2023themis} provided a landmark solution that achieved standard BFT liveness while retaining strong fairness guarantees. Themis inherited the core graph-based methodology but introduced a more sophisticated workflow to solve the weak liveness problem: (1) Collect and Update. Each replica sends not only its locally observed sequence of new transactions (local orders) but also information about pairwise orderings of previously seen but unfinalized transactions (updates). (2) Generate and Classify Dependency Graph. Themis builds a dependency graph from \(n-f\) signed local orders using a non-blank/edge threshold
\(T_{\text{edge}}=T_{\text{non-blank}}=\lfloor n(1-\gamma)+\gamma f+1\rfloor\).
Each transaction is classified per batch by its support:
\emph{solid} if it appears in at least \(n-2f\) local orders,
\emph{blank} if it falls below \(T_{\text{non-blank}}\),
and \emph{shaded} otherwise. (3) Extract Fair Order with Liveness Anchor. To handle Condorcet cycles and ensure progress, the leader identifies SCCs. It then performs a topological sort on the resulting condensation graph. A prefix of the fair order is extracted up to the last SCC that contains at least one solid transaction. This solid transaction acts as a liveness anchor, guaranteeing that the finalized prefix grows over time. The internal ordering of transactions within an SCC is resolved using a deterministic algorithm, such as finding a Hamiltonian cycle. (4) Propose and Symmetrically Verify. The leader broadcasts this ordered list and supporting data in a block. Crucially, to guard against a malicious leader, all other replicas must verify the correctness of this ordering by re-executing the entire, computationally intensive process from the same raw inputs.

While robust and complete, this design cemented two primary performance bottlenecks. The first is the algorithmic bottleneck: building the dependency graph for a batch of $|B|$ transactions requires $\mathcal{O}(|B|^2)$ pairwise weight calculations. The second, and more architecturally critical, is the symmetric verification model. The requirement for every replica to independently re-run the leader's logic tightly couples this expensive ordering computation with the critical path of the underlying BFT consensus. The entire system's progress is gated by this synchronous "order-then-verify" cycle, forcing system-wide, redundant computation and fundamentally limiting throughput and latency.

\textbf{SpeedyFair: Decoupling for Latency Hiding.} The SpeedyFair protocol~\cite{mu2024separation} later made a crucial contribution by directly addressing the architectural bottleneck of tight coupling. Its core insight was to decouple the fair ordering process from the BFT consensus critical path. SpeedyFair proposed a pipelined architecture with a dedicated Optimistic Fair Ordering (OFO) protocol that runs concurrently with the main BFT consensus. In this OFO layer, a "virtual leader" gathers local orders and broadcasts a notification containing the selected inputs. All replicas then compute the fair order for this batch in parallel. The resulting ordered "fragment" is certified with a quorum certificate, which attests to the successful completion of the fair ordering process for that fragment. The main BFT protocol then simply needs to agree on these pre-ordered and certified fragments.

This decoupling effectively hides the ordering latency from the consensus critical path, as the system can be finalizing a previously ordered fragment while the next one is being computed, thus improving end-to-end performance. However, while SpeedyFair elegantly solved the problem of architectural coupling, it did not address the underlying algorithmic and verification inefficiency. The fair ordering mechanism within its OFO process still requires a full, from-scratch reconstruction of the dependency graph in each round. Both the leader and the verifying followers must still perform the same expensive, redundant $\mathcal{O}(|B|^2)$ computation. The fundamental problem of costly, symmetric verification persists, merely relocated to a parallel pipeline.

\begin{figure*}[t]
  \centering
  \includegraphics[width=0.9\textwidth]{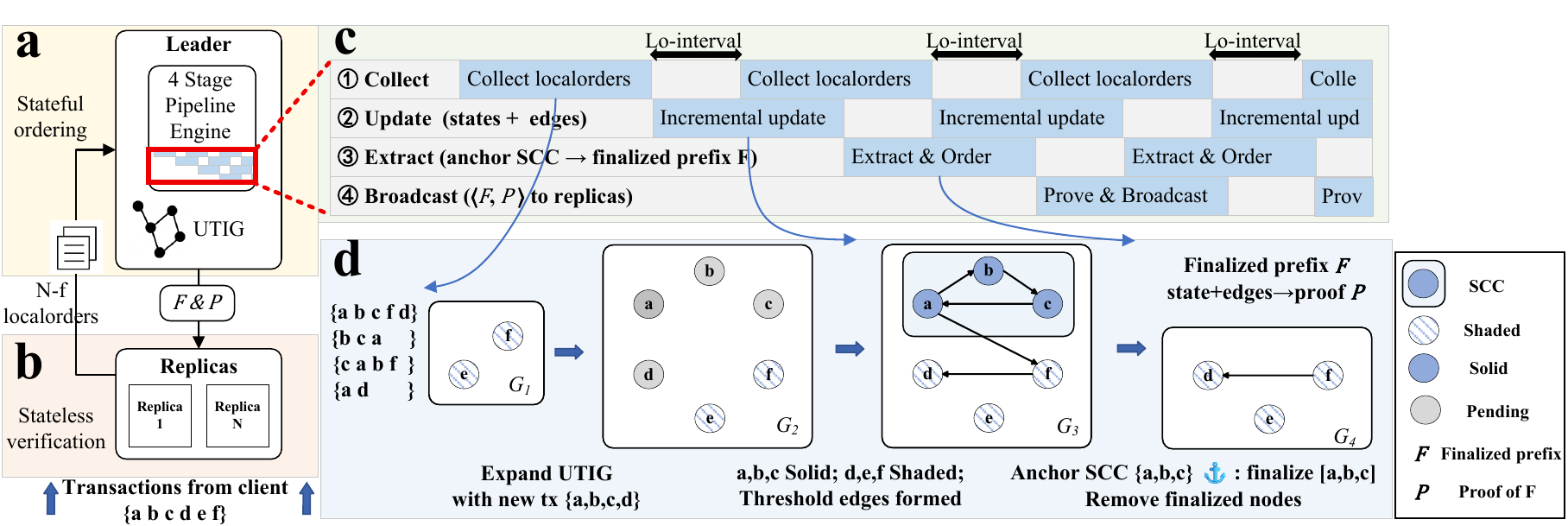}
\caption{Overview of the AUTIG architecture.
\textbf{Panels.} (a) \emph{Leader (stateful ordering)} that maintains the UTIG and executes the pipeline; (b) \emph{Replicas (stateless verification)}—clients first submit transactions to replicas, which then send $n\!-\!f$ signed \emph{LocalOrders} to the leader; (c) the overlapped four–stage pipeline timeline; (d) the corresponding UTIG evolution for this batch.
\textbf{Workflow.} The leader maintains a persistent UTIG across rounds and executes a four-stage pipeline: \ding{172} Collect $n-f$ $LocalOrders$ and add their transactions as \emph{Pending} nodes; \ding{173} Incrementally update supports to classify states ($a,b,c$ become Solid; $d,f,e$ remain Shaded), accumulate weights $W \leftarrow W_{\text{prev}}+W_{\text{batch}}$, and re-orient only affected edges; \ding{174} Extract \& Order by running Tarjan/condensation on the non-blank subgraph and selecting the last SCC with a Solid node as the anchor, yielding a finalizable prefix; \ding{175} Prove \& Broadcast by attaching state and edge-weight proofs, enabling stateless verification at replicas.
}

  \label{fig:autig-overview}
\end{figure*}

\section{System Model and Preliminaries}
\label{sec:model}

\subsection{System and Network Model}

\textbf{System Participants.} We consider a system composed of a static set of $n$ replicas, $\mathcal{R} = \{r_1, \dots, r_n\}$, which together provide a fair ordering service. One replica is designated as the order leader, responsible for proposing the final transaction sequence, while others act as followers. External clients $\mathcal{C}$ submit transactions to all replicas for inclusion.

\textbf{Network Model.} We assume the partially synchronous model~\cite{dwork1988consensus}: after an unknown GST, messages between correct replicas are delivered within a bounded delay~$\Delta$. AUTIG is a modular service that outputs proof-carrying, fairly ordered fragments for a downstream BFT. For the ordering path, after GST, if an honest Order Leader can reach at least $n{-}f$ correct replicas and obtain their signed $LocalOrder$ messages, a valid fragment is produced within a constant number of message exchanges—i.e., progress depends on actual delays rather than on~$\Delta$—and follower verification is a local $O(|P|)$ audit with no extra synchrony assumptions. We assume the hosting stack eventually installs an honest Order Leader; under this assumption AUTIG adds no synchrony requirements beyond Themis~\cite{kelkar2023themis}, and when composed with an optimistic-responsive BFT such as HotStuff~\cite{yin2019hotstuff}, end-to-end commit latency is inherited from the base protocol. We do not assume eventual inclusion of any specific round-$r$ honest $LocalOrder$; precedence evidence not counted in round $r$ can be re-supplied in later rounds via fresh $LocalOrders$, and cumulative weights are across-round sums. 

\subsection{Adversarial Model}
We consider a computationally bounded Byzantine adversary $\mathcal{A}$ that actively seeks to violate the protocol's fairness and safety properties.

\textbf{Byzantine Replicas.}
The adversary controls a set of up to $f$ replicas, denoted $\mathcal{R}_B \subset \mathcal{R}$. The adversary is adaptive, meaning it can choose which replicas to corrupt at any point during the protocol's execution. A corrupted replica $r \in \mathcal{R}_B$ may exhibit arbitrary, malicious behavior by deviating from the protocol in any way. This includes, but is not limited to: (1) Message Manipulation. Sending malformed or conflicting messages to different correct replicas. (2) Selective Communication. Withholding messages or selectively delaying their transmission (within the constraints of the network model). (3) Lying about Observations. Falsifying the reception time or relative order of transactions when constructing its $LocalOrder$ messages. (4) Collusion. Coordinating actions with other corrupted replicas in $\mathcal{R}_B$ to manipulate the outcome of the ordering process.

A replica $r \in \mathcal{R}_C = \mathcal{R} \setminus \mathcal{R}_B$ that is not corrupted is called correct. We assume the adversary cannot break the underlying cryptographic primitives.

\textbf{Network Adversary.} The adversary $\mathcal{A}$ also has partial control over the network fabric. It can reorder and delay messages up to the bounds defined by the partial synchrony model. Formally, for any message $m$ sent by a correct replica $r_i \in \mathcal{R}_C$ at time $t \ge \text{GST}$ to another correct replica $r_j \in \mathcal{R}_C$, the adversary can control its delivery time $t_{delivery}$, but must ensure $t_{delivery} \le t + \Delta$. The adversary cannot partition correct replicas indefinitely after GST. This power allows the adversary to strategically influence the perceived arrival order of transactions at different replicas.

\subsection{Protocol Objectives}
AUTIG is designed as a modular service that produces a stream of ordered transaction batches. When integrated with a secure BFT consensus protocol, the combined system must satisfy BFT properties alongside the core fairness guarantee.

\textbf{Safety and Liveness.} The system must uphold the standard properties of state machine replication: \textbf{Safety:} All correct replicas deliver the same sequence of transaction batches. This is primarily guaranteed by the underlying BFT consensus protocol that consumes AUTIG's output. \textbf{Liveness:} Every valid transaction submitted by a client and received by a sufficient number of correct replicas is eventually included in a batch delivered by all correct replicas.

\textbf{Batch-Order-Fairness.} Following Themis~\cite{kelkar2023themis}, we adopt a batch-based notion parameterized by \(\gamma\in(\tfrac{1}{2},1]\), where \(\gamma\) denotes the fraction of all replicas observing a precedence.
Let the global delivery be partitioned into consecutive batches \(C_1,\dots,C_k\), and write \(b(\mathrm{tx})=j\) iff \(\mathrm{tx}\in C_j\).
For any two transactions \(\mathrm{tx}_1,\mathrm{tx}_2\), if at least a \(\gamma\)-fraction of all replicas report \(\mathrm{tx}_1\prec\mathrm{tx}_2\), the protocol must ensure \(b(\mathrm{tx}_1)\le b(\mathrm{tx}_2)\).
Equivalently, \(\mathrm{tx}_1\) may appear in an earlier batch than \(\mathrm{tx}_2\) or co-appear in the same batch (to resolve Condorcet cycles), but it can never be finalized in a batch strictly after that of \(\mathrm{tx}_2\). We inherit the Themis relation among \((n,f,\gamma)\) required to achieve the above property together with standard liveness under partial synchrony: \(n>\frac{2f(\gamma+1)}{2\gamma-1}\) with \(\tfrac{1}{2}<\gamma\le 1\), which at \(\gamma=1\) yields the familiar Byzantine threshold \(n\ge 4f+1\).
The factor \(2\gamma-1\) captures the honest–honest intersection needed to make majority precedence decisive against up to \(f\) Byzantine votes; we follow Themis’s convention so thresholds and bounds align consistently across the paper.

Batch-order-fairness constrains only inter-transaction precedence; it is deliberately silent about the order within the same SCC/batch. Large SCCs may thus permit intra-batch MEV. To remove leader control, we linearize each SCC by the lexicographic order of $H(\mathrm{txid}\parallel\salt_r)$, where $\salt_r$ is the per-round salt bound in the fragment header; this tie-break leaves Theorem~\ref{thm:verif_sc} intact.

\subsection{Cryptographic Assumptions}
Our protocol's security relies on standard cryptographic primitives and assumptions. \textbf{Digital Signatures.} We assume a Public Key Infrastructure (PKI) where each replica possesses a unique, unforgeable public–private key pair; all protocol messages (e.g., local orders) are digitally signed to ensure authenticity and integrity. \textbf{Cryptographic Hash Function.} We assume a collision-resistant hash \(H(\cdot)\) to generate compact, unique digests for transactions and protocol data; we use domain-separated hashes to avoid cross-purpose collisions and to prevent leader grinding in tie-breaking. Let \(H\) be collision-resistant; we define \(\id{salt}_r \coloneqq H(H_{\mathcal F,\prev}\parallel r\parallel \leaderpk\parallel \domsalt)\) and \(H_{\mathcal F} \coloneqq H(\domfrag\parallel \Enc(F,\mathcal L_{\batch},\mathcal P,H_{\mathcal F,\prev}))\).

\subsection{Foundations of Graph-Based protocol}
AUTIG builds upon the graph-based ordering paradigm. We now formally define the key concepts that constitute the inputs and internal state of our ordering algorithm.

\textbf{Local Order.} The primary input to the ordering process is the local order from each replica. A local order from replica $r_i$, denoted $L_i$, is a signed message containing a sequence of transaction identifiers, $\langle tx_1, tx_2, \dots \rangle$, ordered according to their reception time at $r_i$. This serves as replica $r_i$'s "vote" on the transaction ordering for a given round.

\textbf{Dependency Graph and Cumulative Weights.} The collective ordering preferences of the network are aggregated and modeled as a directed dependency graph $G = (V, E)$. The set of vertices $V$ corresponds to the set of active, unconfirmed transactions. The formation of edges in $E$ is based on cumulative evidence from local orders.

For any pair of transactions $(tx_u, tx_v)$, their cumulative ordering weight, denoted $W(u,v)$, is the total number of local orders in which $tx_u$ precedes $tx_v$, accumulated across all rounds of the protocol. An edge $(u,v) \in E$ exists if the weight $W(u,v)$ surpasses a system-wide threshold, signifying a supermajority agreement on this relative order.

\textbf{Transaction States and Thresholds.} In each round the leader admits exactly \(n{-}f\) signed $LocalOrder$ messages (the first \(n{-}f\) admissible by arrival); if fewer are available, no proposal is formed and the leader keeps collecting. A transaction with per-round support \(c\) (the number of admitted orders that list it) is classified using two public thresholds derived from \((n,f,\gamma)\): \emph{Solid} if \(c\ge T_{\text{solid}}\) with \(T_{\text{solid}}=n-2f\); \emph{Blank} if \(c<T_{\text{non-blank}}\) with \(T_{\text{non-blank}}=\lfloor n(1-\gamma)+\gamma f+1\rfloor\); \emph{Shaded} otherwise. We also set \(T_{\text{edge}}=T_{\text{non-blank}}\) to orient edges and ignore blanks (and their incident edges) in the extraction graph; blank transactions simply do not participate this round and may become non-blank in later rounds as support accrues. Intuition in brief: \(T_{\text{solid}}=n-2f\) guarantees at least \(n-3f\ge f+1\) honest supporters, yielding a robust liveness anchor that cannot be suppressed by \(f\) Byzantine replicas, while \(T_{\text{non-blank}}=T_{\text{edge}}\) ensures any oriented precedence has honest support that dominates the worst-case mix of Byzantine votes and minority honest reversals, thus enforcing the “no-later-than” constraint \(b(tx_1)\le b(tx_2)\). For brevity we write \(\textsc{DetermineState}(c)\in\{\Solid,\Shaded,\Blank\}\) for the mapping defined above; states are recomputed per round from that batch’s supports only (no cross-round carry-over), and we use \(T_{\text{edge}}=T_{\text{non-blank}}\) when orienting edges.

\subsection{Verifiable Fragments and Proofs}
\label{sec:model_fragments}

To facilitate a decoupled and asymmetric architecture, the communication between the order leader and followers is encapsulated in a well-defined data structure: the $VerifiableFairOrderFragment$. This structure serves not only as the proposal unit but also as the foundation for stateless verification and secure state handoff.

A fragment, denoted $\mathcal{F}$, is a tuple $\mathcal{F} = \langle H_{\mathcal{F}}, F, \mathcal{L}_{\text{batch}}, \mathcal{P} \rangle$, containing: $F$. The $FinalOrder$, a proposed list of transaction IDs whose global order is being finalized in this fragment. $\mathcal{L}_{\text{batch}}$: The batch of $n-f$ signed $LocalOrders$ that were used as input to generate this proposal. This provides the raw evidence for the current round. 
$\mathcal{P}$: The $ProofData$, a structured proof containing all assertions needed for a follower to verify $F$ without historical state. It consists of: \textbf{(i) State assertions} for every $y\in F$; \textbf{(ii) Internal-pair totals} for each unordered $\{u,v\}\subseteq F$, i.e., $(W(u,v),W(v,u))$, to verify all dependencies inside $F$; and \textbf{(iii) Frontier completeness}. Let $B_r=\{x:\textsc{DetermineState}(x;\mathcal{L}_{\batch})\neq\Blank\}$ be the set of non-blank transactions in the current batch. The proof lists totals for every boundary pair $(x,y)\in(B_r\setminus F)\times F$; followers check that none of these pairs satisfies the public edge predicate into $F$ (i.e., $\neg P(x\!\to\!y)$ for all), which certifies that $F$ is \emph{down-closed} and thus safely finalizable.
 $H_{\mathcal{F}}$: A cryptographic digest of the canonical representation of $\langle F, \mathcal{L}_{\text{batch}}, \mathcal{P} \rangle$. This digest is the object on which the underlying BFT protocol reaches consensus. It ensures the integrity and tamper-evidence of the entire fragment once committed to the log. By committing $H_{\mathcal{F}}$, the system creates an immutable, universally agreed-upon record of not just the ordering decision ($F$), but also its complete justification ($\mathcal{L}_{\text{batch}}, \mathcal{P}$).

\section{The AUTIG Protocol}
\label{sec:protocol}

This section details the design of AUTIG, a modular service that provides high-performance, verifiable order-fairness. AUTIG builds upon the theoretical foundations of graph-based fairness established by Themis~\cite{kelkar2023themis} but introduces a fundamentally different, asymmetric architecture designed for high throughput and scalability. It achieves this by strategically decoupling the stateful, computationally intensive ordering task from the lightweight, stateless verification task. An overview of the AUTIG architecture is shown in Fig.~\ref{fig:autig-overview}.

\subsection{Overall Design}

The central innovation of AUTIG is its asymmetric state and computation model. We partition the system's operation into two distinct but coordinated engines: a stateful ordering engine exclusive to the leader, and lightweight verification engines at all follower replicas.

\textbf{The Leader's Stateful Ordering Engine:} The leader maintains a persistent, long-lived data structure we call the Unconfirmed-Transaction Incremental Graph. The UTIG serves as the stateful heart of the system, dynamically tracking the ordering dependencies of all non-finalized transactions. Its purpose is to amortize the cost of graph construction across many rounds, avoiding the expensive process of rebuilding the dependency graph from scratch in each proposal. In each round, the leader incrementally updates this graph using a new batch of $LocalOrders$ from replicas. It then executes a deterministic algorithm to extract a prefix of transactions whose order can be safely finalized. It also generates a compact, self-contained proof that cryptographically attests to the correctness of its ordering decision.

\textbf{The Replicas' Stateless Verification Engine:} Follower replicas run a lightweight verification engine. Replicas do not maintain a persistent copy of the leader's complex dependency graph. Instead, they perform efficient, stateless validation using the self-contained proof generated by the leader. The proof provides the necessary context for followers to audit the leader's actions without requiring any historical state. This transforms verification from a costly $\mathcal{O}(|B|^2)$ re-computation into a fast, resource-efficient audit, significantly reducing the storage and computation burden on followers and simplifying their recovery logic.

This architecture is realized through a continuous, pipelined workflow. The leader's engine is composed of concurrent stages that handle collecting inputs, updating the stateful graph, and extracting orders. This design allows AUTIG to operate as a pluggable service. The leader's engine produces a stream of fragments, which are then consumed and finalized by any underlying BFT consensus protocol after lightweight verification by the followers.

\subsection{The UTIG}
\label{sec:utig_definition}

To efficiently compute and maintain the state of the logical dependency graph over time, the leader utilizes the UTIG. The UTIG is a concrete data structure that materializes the abstract concepts from Section~\ref{sec:model}. Formally, the UTIG, denoted $G_{\text{utig}}$, is a tuple $(V, W, \text{States}, E, R)$, where $V$: The set of vertices, representing all transaction IDs currently being tracked. $W$: A map storing the cumulative ordering weight $W(u,v)$ for each transaction pair. This is the physical storage for the weights defined conceptually in the model. $\text{States}$: A map storing the current state for each transaction in $V$, as determined by its support in the most recent batch.  $E$: The set of directed dependency edges. An edge $(u,v) \in E$ exists if and only if both transactions are non-blank and a deterministic weight-based rule is met. This set is explicitly maintained and updated. $R$: A map storing the round number in which each transaction was last seen. This metadata is specific to AUTIG's implementation and is used for garbage collection (pruning).

The orientation of edges in \(E\) is governed by a deterministic and anti-symmetric predicate \(P(u\to v)\): \(P(u\to v)\equiv [W(u,v)\ge T_{\text{edge}}]\land([W(v,u)<T_{\text{edge}}]\ \lor\ [W(u,v)>W(v,u)]\ \lor\ ([W(u,v)=W(v,u)]\land u<v))\).
This rule ensures that if a supermajority preference exists, an edge is formed. In cases of mutual supermajority, the edge points from the transaction with stronger support, with ties broken deterministically. The leader can perform incremental updates to the graph, amortizing the computational cost across many rounds instead of rebuilding it from scratch.

\subsection{Protocol Workflow and Algorithms}

The AUTIG protocol proceeds in a continuous, pipelined fashion. We describe the key stages and the distinct roles of the leader and replicas.

\subsubsection{Phase 1: Local Order Generation and Collection}

\textbf{Replica's Role:} Each correct replica $i$ continuously receives transactions from clients into a local mempool. It timestamps and orders them by their arrival time. Periodically, or upon a request from the leader for a specific round $r$, it constructs a $LocalOrder$ message, $L_i$. This message contains a sequence of transaction IDs from its mempool, ordered locally, and is tagged with the round number $r$. The replica then signs the entire message and sends it to the current leader.

\textbf{Leader's Role (Collector Stage):} The leader's collector stage aggregates signed $LocalOrder$ messages for the current round. To form a batch, the leader waits until it has received at least $n-f$ admissible messages. Let $S_r$ be the set of all admissible $LocalOrders$ received for round $r$. The leader then deterministically constructs the batch $\mathcal{L}_{\text{batch}}$ by selecting exactly $n-f$ messages from $S_r$ based on the lexicographical order of the signing replicas' public keys (or identifiers). Any admissible messages not selected are deferred to the next round.

Before selection, the leader enforces a set of critical admissibility checks on each incoming $L_i$: \textbf{Signature Verification.} The digital signature must be valid. \textbf{Round Validity:} The round number must match the current round $r$. \textbf{Uniqueness:} At most one valid $LocalOrder$ per replica per round is considered.

This deterministic selection process, based on stable replica identifiers, removes any ``first-arrival'' advantage present in prior work. An adversary cannot influence which honest replicas are included in the batch by manipulating network delivery times. While the set of included honest replicas may vary slightly from round to round, any batch of size $n-f$ is guaranteed to contain at least $n-2f$ honest inputs, and the protocol's use of cumulative weights across rounds ensures that fairness properties are preserved in the long run.

\subsubsection{Phase 2: Incremental UTIG Update}
This phase is executed exclusively by the leader's updater stage. This process is detailed in Algorithm~\ref{alg:incremental_update}. The update unfolds in three logical steps.

\textbf{Step 1: State Update.} 
The first step is to determine how the new batch of observations affects the status of each transaction. The leader iterates through every transaction mentioned in $\mathcal{L}_{\text{batch}}$ and calculates its support count—the number of local orders in the batch that contain it. Based on this count, it re-evaluates each transaction's state according to the public thresholds $T_{\text{solid}}$ and $T_{\text{non-blank}}$. A transaction's state can change, or remain the same. The leader meticulously tracks the set of transactions whose states have changed, which we call $DirtyNodes$. 

\textbf{Step 2: Cumulative Weight Update.}
Let $W_{\batch}(u,v)$ be the number of local orders in \(\mathcal{L}_{\batch}\) where \(u\) precedes \(v\). The cumulative weights are append-only across rounds: \(W(u,v)\leftarrow W(u,v)+W_{\batch}(u,v)\).
The leader updates the cumulative evidence for pairwise orderings. It processes each $LocalOrder$ $L_i$ in the batch. For every ordered pair of transactions $(u, v)$ appearing in $L_i$ (where $u$ comes before $v$), the leader increments the cumulative weight $W(u,v)$ stored in its UTIG. During this process, the leader also identifies $DirtyPairs$—pairs of transactions for which the weight update has caused one of their directional weights to cross the $T_{\text{edge}}$ threshold for the first time. 

\textbf{Step 3: Optimized Edge Recomputation.}
The final step is to update the set of directed edges $E$ in the UTIG. A naive approach would be to re-evaluate all $\mathcal{O}(|V|^2)$ possible edges, nullifying the benefits of the incremental model. Instead, AUTIG uses a highly optimized procedure. It computes a minimal set of $AffectedPairs$ that require re-evaluation. This set is the union of the $DirtyPairs$ and all pairs that are adjacent to any of the $DirtyNodes$. For each pair $\{u,v\}$ in this small set, the leader first removes any existing edge between them and then re-applies the deterministic predicate $P(\cdot \to \cdot)$. If the updated weights and states now satisfy the predicate for either $u \to v$ or $v \to u$, a new directed edge is added. Edges for pairs outside this affected set remain unchanged, dramatically reducing the computational work compared to a full graph reconstruction.

\begin{algorithm}[h!]
\caption{Incremental UTIG Update}
\label{alg:incremental_update}
\footnotesize
\Input{Batch $\mathcal{L}_{\text{batch}}$, round $r$; UTIG $G_{\text{utig}}{=}(V,W,S,E,R)$}
\Output{Updated $G_{\text{utig}}$}
$DirtyNodes\!\gets\!\varnothing$;\; $DirtyPairs\!\gets\!\varnothing$\;
$Supp\!\gets\!\textsc{CountBatchSupport}(\mathcal{L}_{\text{batch}})$\;
\tcp{CountBatchSupport returns per-tx support counts in the \emph{current} batch}
\ForEach{$(tx,c)\in Supp$}{
  \textsc{AddIfAbsent}$(V,tx)$;\; $R[tx]\!\gets\!r$\;
  $old\!\gets\!S[tx]$;\; $S[tx]\!\gets\!\textsc{DetermineState}(c)$\;
  \If{$S[tx]\!\neq\!old$}{$DirtyNodes.\textsc{add}(tx)$}
}
$W_{\batch}\!\gets\!\textsc{CountBatchPairWeights}(\mathcal{L}_{\text{batch}})$\;
\ForEach{pair $(u,v)$ with $W_{\batch}(u,v)>0$}{
  $old\!\gets\!W(u,v)$;\;
  $W(u,v)\!\gets\!old+W_{\batch}(u,v)$\;
  \If{$old\!<\!T_{\text{edge}}$ \textbf{and} $W(u,v)\!\ge\!T_{\text{edge}}$}{
    $DirtyPairs.\textsc{add}(\{u,v\})$
  }
}
$V_{\text{nb}}\!\gets\!\{\,x\in V: S[x]\neq \Blank\,\}$\;
$Affected\!\gets\!DirtyPairs\;\cup\;\{\{u,v\}\!:\,u\!\in\!DirtyNodes,\,v\!\in\!V_{\text{nb}},\,u\!\neq\!v\}$\;
\ForEach{unordered $\{u,v\}\in Affected$}{
  $E.\textsc{remove}(u\!\to\!v)$;\; $E.\textsc{remove}(v\!\to\!u)$\;
  \If{$S[u]\!=\!\Blank$ \textbf{or} $S[v]\!=\!\Blank$}{\textbf{continue}}
  \tcp{Orientation predicate $P(\cdot)$ from §\ref{sec:model}}
  \If{$P(u\!\to\!v)$}{$E.\textsc{add}(u\!\to\!v)$}
  \ElseIf{$P(v\!\to\!u)$}{$E.\textsc{add}(v\!\to\!u)$}
}
\end{algorithm}

\subsubsection{Phase 3: Order Extraction and Proof Generation}
\label{sec:extract_and_prove_details}
After updating the UTIG, the leader's proposer stage attempts to extract a new prefix of transactions whose order can be safely finalized. This multi-step process, formalized in Algorithm~\ref{alg:extract_and_prove}, is designed to ensure both fairness and liveness, and culminates in the generation of a self-contained proof.

\textbf{Step 1: Constructing the Extraction Subgraph.}
The leader begins by creating a temporary graph, $G_{\text{extract}}$, which is a subgraph of the UTIG. This subgraph contains the non-blank nodes and the directed edges between them. This step focuses the subsequent, computationally intensive graph algorithms solely on the set of transactions that are currently relevant for ordering, filtering out those with insufficient support.

\textbf{Step 2: Identifying Ordering Dependencies and Cycles.}
The leader then analyzes the topology of $G_{\text{extract}}$ to understand the collective ordering preferences. It runs Tarjan's algorithm to identify all Strongly Connected Components. The leader then constructs the condensation graph, $G_{\text{condensed}}$, an acyclic graph where each node represents an entire SCC from $G_{\text{extract}}$. The edges in this new graph represent the dependencies between these cycles.

\textbf{Step 3: Establishing a Partial Order.}
By performing a topological sort on the acyclic condensation graph, the leader obtains a linear ordering of the SCCs. This represents a valid partial order of transaction groups that respects all non-cyclical supermajority preferences.

\textbf{Step 4: The Solid Anchor Principle for Liveness.}
The leader inspects the topological sort to find the last SCC in the sequence that contains at least one transaction in the \(\Solid\) state. This SCC is designated as the solid anchor. The protocol commits to finalizing all transactions up to and including this anchor. The presence of a solid transaction guarantees that the finalized prefix of the log will grow over time, as solid transactions are, by definition, widely observed and their state is stable. If no solid anchor is found, no transactions can be finalized in this round.

\textbf{Step 5: Extracting and Linearizing the Final Order Prefix.}
The leader now constructs the final ordered list of transactions, $F$. It traverses the topological sort from the beginning up to the solid anchor. For each SCC encountered: If the SCC contains a single transaction, that transaction is appended to $F$. For multi-transaction SCCs, we apply a leader-independent, hash-based linearization: sort members by the lexicographic order of $H(\id{txid}\,\|\,\id{salt}_r)$.  This approach is consistent for all replicas and satisfies the requirements of batch-order-fairness.

\textbf{Step 6: Generating the Proof of Fairness.} If a non-empty prefix $F$ is extracted, the leader emits $\mathcal{P}$ as specified in Section~\ref{sec:model_fragments}: state assertions for all $y\in F$, internal-pair totals for every $\{u,v\}\subseteq F$, and the frontier-completeness assertions covering all $(B_r\setminus F)\times F$. This self-contained proof enables fast, stateless auditing that $F$ is both fairly ordered and safely finalizable.

\begin{algorithm}[h!]
\caption{Extract Fair Prefix and Generate Proof}
\label{alg:extract_and_prove}
\footnotesize
\Input{UTIG $G_{\text{utig}}{=}(V,W,S,E,R)$; batch $\mathcal{L}_{\text{batch}}$}
\Output{Final order $F$, proof $\mathcal{P}$}
$V_{\text{nb}}\!\gets\!\{\,tx\in V: S[tx]\neq \Blank\,\}$;\;
$G_{\text{ext}}\!\gets\!(V_{\text{nb}}, E|_{V_{\text{nb}}})$\;
$SCCs\!\gets\!\textsc{Tarjan}(G_{\text{ext}})$;\;
$(G_c,\text{sccOf})\!\gets\!\textsc{Condensation}(SCCs)$;\;
$Topo\!\gets\!\textsc{TopoSort}(G_c)$\;
$\mathrm{anchor}\!\gets\!\textsc{LastIndexWithSolid}(Topo,SCCs,S)$\;
\If{$\mathrm{anchor}\!=\!\bot$}{\KwRet $(\varnothing,\bot)$}
$F\!\gets\![\,]$\;
\For{$i$ over $Topo[1..\mathrm{anchor}]$}{
  $C\!\gets\!SCCs[i]$\;
  \eIf{$|C|\!=\!1$}{$F.\textsc{append}(C[1])$}{$F.\textsc{append}(\textsc{DeterministicLinearize}(C,\salt_r))$}
}
$\mathcal{P}.\text{states}\!\gets\!\{(tx,S[tx])\,:\,tx\in F\}$\;
$\mathcal{P}.\text{infix}\!\gets\!\{(u,v,W(u,v),W(v,u))\,:\,u\neq v,\,u,v\in F\}$\;
$Supp_{\batch}\!\gets\!\textsc{CountBatchSupport}(\mathcal{L}_{\batch})$\;
$B_r\!\gets\!\{\,x:\ \textsc{DetermineState}(Supp_{\batch}[x])\neq \Blank\,\}$\;
$X\!\gets\!(B_r\setminus F)$;\; \tcp{active non-blank outside the prefix}
\tcp{emit the full Cartesian product in a canonical order}
\ForEach{$x$ in \textsc{SortByTxid}$(X)$}{
  \ForEach{$y$ in \textsc{SortByTxid}$(F)$}{
    $\mathcal{P}.\text{frontier}.\textsc{append}\big(x,y,\,W(x,y),\,W(y,x)\big)$\;
  }
}
\KwRet $(F,\mathcal{P})$
\end{algorithm}

\subsubsection{Phase 4: Proposal Broadcast and Verification}
\label{sec:protocol_verification}

\textbf{Leader's Role:}
The leader assembles the final proof-carrying fragment \(\mathcal F=\langle \hdr,\,F,\,\mathcal L_{\batch},\,\mathcal P\rangle\) with header \(\hdr=\langle r,\,\leaderpk,\,H_{\mathcal F,\prev},\,\salt_r,\,H_{\mathcal F}\rangle\) and computes the chained commitment \(H_{\mathcal F}=\Commit\!\big(F,\,\mathcal L_{\batch},\,\mathcal P,\,H_{\mathcal F,\prev};\,\domfrag\big)\) and broadcasts \(\mathcal F\).
The header fixes \(\salt_r\) and binds the fragment to its predecessor; replicas vote only after verifying the header and the full fragment.

\textbf{Follower's Role:}
 Upon receiving a fragment $\mathcal{F}$, a follower replica performs a fast, stateless verification. It does not consult any local historical state; all the information it needs is contained within the fragment itself. 

\textbf{Step 1: Integrity and Consistency Checks.}
The follower verifies the integrity of the fragment by re-computing its digest and ensuring it matches $H_{\mathcal{F}}$. It checks the consistency of the state assertions. For each transaction in $F$, it computes the support count from the provided $\mathcal{L}_{\text{batch}}$ and verifies that the resulting state matches the state asserted in $\mathcal{P}$.

\textbf{Step 2: Auditing Weights and Edge Correctness.}
The follower then audits the core of the fairness proof: the weight assertions. For each pair of transactions $\{u,v\}$ covered by the proof, it first calculates the weight contribution from the current batch, $W_{\text{batch}}(u,v)$, by counting occurrences in $\mathcal{L}_{\text{batch}}$. It then subtracts this batch weight from the total cumulative weight asserted in the proof, $w_{uv}$. The result is the implied historical weight. Check 1 (Non-Negative History): This implied history must be non-negative. A negative value would mean the leader fabricated historical support that never existed, so the proof is immediately rejected. Check 2 (Edge Direction): The follower uses the asserted total weights $w_{uv}$ and $w_{vu}$ to re-evaluate the edge predicate $P(\cdot \to \cdot)$. It verifies that the resulting edge direction (or lack thereof) is consistent with the graph structure implied by the proof. Check 2' (Frontier Soundness): The follower performs this check for the boundary pairs as well. It verifies that for every active transaction $x \in B_r \setminus F$, the predicate $P(x \to y)$ is false for all $y \in F$. This confirms that the prefix $F$ is down-closed and no unfinalized transactions have a valid dependency pointing into it.

\textbf{Step 3: Reconstructing and Validating the Extraction Logic.}
Having verified the soundness of the proof's underlying data, the follower performs the final check. It uses the now-trusted state and weight assertions from $\mathcal{P}$ to construct its own temporary, in-memory verification subgraph, $G_{\text{verify}}$. This graph is guaranteed to be identical to the one the leader used for extraction. The follower then re-runs the \textit{exact same deterministic extraction logic} on $G_{\text{verify}}$ (Tarjan's, condensation, topological sort, solid anchor identification, and deterministic linearization). The resulting final order, $F'$, must be identical to the order $F$ proposed by the leader. If it matches perfectly, the proof is valid. Only if all these checks pass does the follower consider the fragment's ordering fair and vote to accept it in the consensus layer. 

\begin{algorithm}[h!]
\caption{Stateless Verification at Follower}
\label{alg:verification}
\footnotesize
\Input{$\mathcal{F}=\langle \hdr,F,\mathcal{L}_{\batch},\mathcal{P}\rangle$, 
$\hdr=\langle r,\leaderpk,H_{\mathcal F,\prev},\salt_r,H_{\mathcal F}\rangle$}
\Output{accept / reject}

\If{$\mathcal{P}=\bot$}{\KwRet $(|F|=0)$}

\If{$\salt_r \neq H(H_{\mathcal F,\prev}\parallel r\parallel \leaderpk\parallel\domsalt)$ \textbf{or}
    $H_{\mathcal F}\neq H(\domfrag\parallel \Enc(F,\mathcal{L}_{\batch},\mathcal{P},H_{\mathcal F,\prev}))$}{\KwRet reject}

$Supp\gets\textsc{CountBatchSupport}(\mathcal{L}_{\batch})$;\;
$W_{\batch}\gets\textsc{CountBatchPairWeights}(\mathcal{L}_{\batch})$\;

\For{$(tx,\hat s)\in \mathcal{P}.\text{states}$}{
  \If{$\textsc{DetermineState}(Supp[tx])\neq \hat s$}{\KwRet reject}
}

$B_r\gets\{x:\textsc{DetermineState}(Supp[x])\neq \Blank\}$;\;
\If{$|\mathcal{P}.\text{infix}|\neq |F|\cdot(|F|-1)$}{\KwRet reject}
\If{$|\mathcal{P}.\text{frontier}|\neq |(B_r\setminus F)|\cdot|F|$}{\KwRet reject}

\For{$(a,b,w_{ab},w_{ba})\in \mathcal{P}.\text{infix}\cup\mathcal{P}.\text{frontier}$}{
  \If{$w_{ab}<W_{\batch}(a,b)$ \textbf{or} $w_{ba}<W_{\batch}(b,a)$}{\KwRet reject}
}

$G_v\gets$ graph on $F$ with edge $u\!\to\!v$ iff $P(u\!\to\!v;w_{uv},w_{vu})$ for $(u,v,\cdot,\cdot)\in\mathcal{P}.\text{infix}$\;

\For{$(x,y,w_{xy},w_{yx})\in \mathcal{P}.\text{frontier}$}{
  \If{$x\notin B_r$ \textbf{or} $y\notin F$}{\KwRet reject}
  \If{$P(x\!\to\!y;w_{xy},w_{yx})$}{\KwRet reject}
}

$SCCs\gets\textsc{Tarjan}(G_v)$;\;
$Topo\gets\textsc{TopoSort}(\textsc{Condensation}(SCCs))$;\;
$\mathrm{anchor}\gets\textsc{LastIndexWithSolid}(Topo,SCCs,\mathcal{P}.\text{states})$;\;
\If{$|F|>0$ \textbf{and} $\mathrm{anchor}=\bot$}{\KwRet reject}

$F'\gets[\,]$;\;
\For{$i$ over $Topo[1..\mathrm{anchor}]$}{
  $C\gets SCCs[i]$;\;
  $F'.\textsc{append}(|C|=1?\,C[1]:\textsc{DeterministicLinearize}(C,\salt_r))$;}

\KwRet $(F'=F)$
\end{algorithm}

\subsubsection{Consensus and Finalization}

If a follower's verification of $\mathcal{F}$ succeeds, it participates in the underlying BFT protocol (e.g., HotStuff) by voting affirmatively on the leader's proposal containing $\mathcal{F}$. If the BFT protocol reaches consensus on the proposal, all correct replicas deterministically commit the final order $F$ to their state machines for execution. The leader also removes the now-finalized transactions from its UTIG to maintain its bounded size. To ensure the UTIG remains bounded over long-term operation, AUTIG incorporates a secure pruning mechanism, detailed in Appendix~\ref{app:pruning}, which is designed to be invisible to the stateless verification process.

\subsubsection{Dual leadership}
To preserve the amortized benefits of this incremental graph, AUTIG decouples the \emph{Order Leader}, which maintains the UTIG and produces proof-carrying fragments, from the \emph{BFT consensus leader}, which drives agreement on fragment digests. The Order Leader is intended to be long-lived and changes only upon verifiable misbehavior. Decoupling keeps ordering progress independent of consensus view changes and ensures that a state handoff is invoked only when fairness is directly compromised.

\textbf{Fault detection.} Stateless verification (Alg.~\ref{alg:verification}) is deterministic. A follower that rejects a fragment $\mathcal F$ holds a non-repudiable complaint (the committed digest $H_{\mathcal F}$ plus the leader’s signature). Collecting $f{+}1$ complaints for the same digest triggers an order-leader handoff.

\textbf{State recovery.} The new Order Leader reconstructs UTIG state from the last committed fragment digest by fetching the full fragment, checking its hash, and replaying the proof to re-initialize nodes, states, and cumulative weights; a one-shot recovery batch re-synchronizes the frontier. Full protocol details and proofs are in Appendix~\ref{app:handoff}.

\subsection{Security and Liveness Analysis}
\label{sec:analysis}
We prove that AUTIG satisfies order-fairness and BFT safety/liveness under the system model of \ref{sec:model}, relying on Algorithms~\ref{alg:incremental_update}--\ref{alg:verification} and the public thresholds and edge predicate $P(\cdot)$ defined therein. A complexity analysis is provided in Appendix~\ref{app:complexity}.

\subsection{Safety and Order-Fairness}
We first show that any accepted fragment is exactly the canonical fair prefix and that any canonical fair prefix is accepted. The crux is that the proof carried in a fragment lets followers losslessly reconstruct the leader’s extraction view.

\begin{theorem}[Soundness and Completeness of Stateless Verification]
\label{thm:verif_sc}
Algorithm~\ref{alg:verification} accepts a fragment $\mathcal F=\langle \hdr,F,\mathcal L_{\batch},\mathcal P\rangle$ iff $F$ is the unique output of the canonical extractor (Tarjan $\to$ condensation DAG $\to$ topological order $\to$ solid-anchor cutoff $\to$ deterministic intra-SCC linearization) on the fair-order graph induced by cumulative weights and the current batch; hence $F$ satisfies $\gamma$-batch-order fairness.
\end{theorem}

\begin{lemma}[Proof-Induced Graph Equivalence]
\label{lem:graph_equiv}
If $\mathcal F$ passes Algorithm~\ref{alg:verification}’s integrity checks (header binding, per-tx state recomputation from $\mathcal L_{\batch}$, pairwise weight non-negativity, and frontier completeness), then the follower’s verification graph is equivalent to the leader’s extraction view for the purpose of prefix extraction: inside $F$, orienting each pair $\{u,v\}$ by applying $P(\cdot)$ to the asserted totals $(w_{uv},w_{vu})$ yields the same internal edges as in the leader’s cumulative graph, and the frontier check certifies that for all $(x,y)\in(B_r\setminus F)\times F$ we have $\neg P(x\!\to\!y)$, so $F$ is down-closed among non-blank transactions.
\end{lemma}

\begin{proof}[Proof of Theorem~\ref{thm:verif_sc}]
(\emph{Soundness}) If the verifier accepts, all integrity checks pass. By Lemma~\ref{lem:graph_equiv}, the verifier’s reconstruction coincides with the leader’s extraction view on $F$ and certifies that $F$ is down-closed w.r.t.\ active non-blank transactions. Running the deterministic extractor on the identical graph yields exactly $F$; SCC contiguity up to the last $\Solid$-containing SCC enforces $\gamma$-batch-order fairness.

(\emph{Completeness}) If $F$ is the canonical output on the true batch graph, the honest leader includes in $\mathcal P$ the per-tx states for $y\in F$, all internal-pair totals for $\{u,v\}\subseteq F$, and the full frontier $(B_r\setminus F)\times F$ with totals. Each total equals (history $+$ recomputed batch counts), so state checks, non-negativity, and frontier down-closure all pass; re-extraction reproduces $F$, hence the verifier accepts.
\end{proof}

\subsection{Liveness and State Boundedness}
We next show that AUTIG makes progress under partial synchrony and keeps the leader’s state bounded; the latter underpins stable throughput.

\begin{theorem}[Liveness with Explicit Bounds]
\label{thm:liveness_bounds}
Assume partial synchrony with delay $\Delta$ after GST and that the order leader waits a $\Delta$-long collection window before forming a batch from the first $n{-}f$ admissible $LocalOrder$s. Let $T_{\round}=T_{\collect}+T_{\update}+T_{\extract}+T_{\broadcast}$. For any valid $tx$ received by at least $n{-}2f$ correct replicas after GST, $tx$ becomes $\Solid$ within $2\Delta+T_{\collect}$ and is finalized within
 \(T_{\final}\le 2\Delta+(1+R_{\defer})\,T_{\round}+N_{\bad}\bigl(T_{\detect}+T_{\handoff}+T_{\recovery}\bigr)\).
\end{theorem}

\begin{proof}[Sketch]
After GST, propagation to $n{-}2f$ correct replicas and batching imply solidity by the next extraction; the solid-anchor rule finalizes a non-empty prefix per round with at most $R_{\defer}$ deferrals. Faulty leaders are detected by Theorem~\ref{thm:verif_sc} and replaced; at most $f$ appear consecutively under round-robin. Full proof in Appendix~\ref{app:sec6proofs}.
\end{proof}

\begin{theorem}[Pruning Safety and Live-Core Preservation]
\label{thm:pruning_safety}
Under the admissible pruning rules in Appendix~\ref{app:pruning}, the extracted prefix up to the solid anchor matches that from an unpruned UTIG; verification remains fragment-local; pruned transactions can be safely re-integrated upon reappearance.
\end{theorem}

\noindent\textit{Proof.} See Appendix~\ref{app:sec6proofs}. 

\begin{theorem}[Recovery Soundness and Convergence with Byzantine Replies]
\label{thm:recovery}
Given the last committed fragment digest $H_{\mathcal F_{\text{anchor}}}$ and any full fragment consistent with it, the new honest order leader reconstructs a UTIG that preserves fairness; historical weights are exact, and a recovery batch from any $\ge n{-}f$ admissible replies converges to the on-line state while rejecting malformed/duplicate/cross-round replies.
\end{theorem}

\noindent\textit{Proof.} See Appendix~\ref{app:sec6proofs}. 

\subsection{Correctness of Performance Optimizations}
\begin{theorem}[Incremental Equivalence and Pipeline Linearizability]
\label{thm:incremental_equiv}
Let $\widetilde G_r$ be the state after Algorithm~\ref{alg:incremental_update} and $G_r^{\star}$ the from-scratch state up to round $r$. Then $\widetilde G_r\equiv G_r^{\star}$; with SWMR on UTIG, the concurrent pipeline is linearizable with the serialization point at the end of Algorithm~\ref{alg:incremental_update}.
\end{theorem}

\noindent\textit{Proof.} See Appendix~\ref{app:optim_proofs}. 

\begin{figure*}[!t]
    \centering
    \begin{subfigure}[b]{0.49\textwidth}
        \centering
        \includegraphics[width=\linewidth]{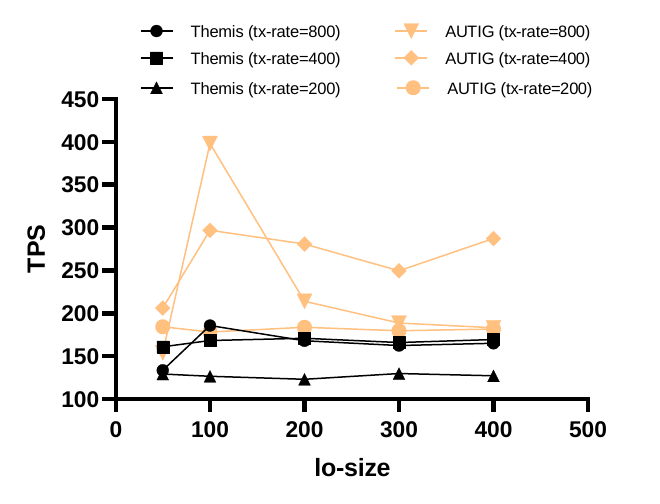}
        \caption{TPS vs. $lo-size$}
        \label{fig:tps}
    \end{subfigure}
    \hfill
    \begin{subfigure}[b]{0.49\textwidth}
        \centering
        \includegraphics[width=\linewidth]{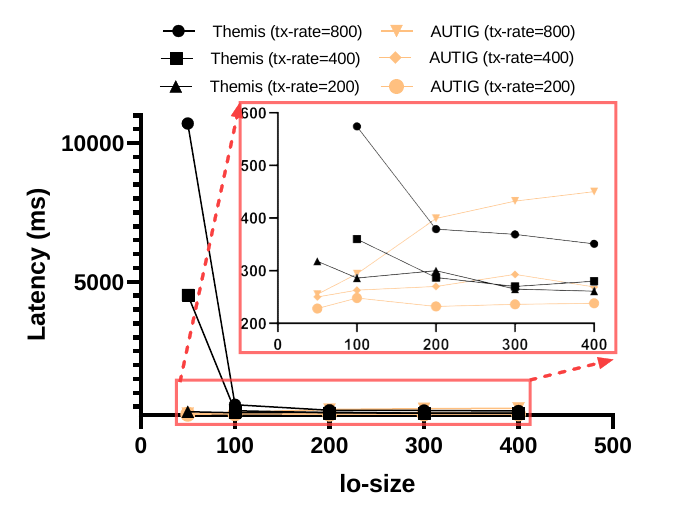}
        \caption{Latency vs. $lo-size$}
        \label{fig:latency}
    \end{subfigure}
    \caption{TPS and latency of AUTIG and Themis under varying $lo-size$ and $tx-rate$.}
    \label{fig:comparison}
\end{figure*}

\section{Evaluation}
\label{sec:evaluation}
Our work builds upon the architectural decoupling paradigm that separates the fair ordering process from the underlying BFT consensus. Our core innovation lies within the ordering layer itself. We evaluate AUTIG and compare it against Themis~\cite{kelkar2023themis}. This choice isolates architectural effects: both AUTIG and Themis implement the same graph-based fairness semantics and extraction pipeline, so any performance delta stems from AUTIG's asymmetric design, not from differences in the fairness algorithm. We do not report a separate SpeedyFair\cite{mu2024separation} baseline because, while SpeedyFair decouples ordering from consensus to hide latency, its fair-ordering engine retains the symmetric, from-scratch graph reconstruction and verification cost of Themis; when consensus-layer effects are factored out, its ordering dataplane reduces to Themis-like behavior. Accordingly, Themis serves as the canonical baseline for ordering throughput and latency.

\subsection{Experimental Setup}

\textbf{Implementation and Platform.}
We implemented a prototype of AUTIG in Go. All experiments were conducted on the AWS EC2 platform, using a network of dedicated instances for replicas and a separate set of client instances. The clients generate a saturated transaction workload to measure the systems' peak performance. Each transaction consists of a unique 256-byte payload; to optimize network bandwidth, protocol messages such as $LocalOrders$ transmit only the 32-byte hash of each transaction.

\textbf{System and Protocol Parameters.}
All experiments are run with a network of $n=21$ replicas tolerating $f=5$ Byzantine faults. We primarily consider the case where the Order Leader is correct. Per protocol rules, a leader must collect at least $n-f=16$ signed LocalOrders to form a batch. A transaction's state and the dependency graph structure are governed by public vote thresholds: it becomes \emph{non-blank} only if it appears in at least $\theta_{\text{nb}} = \lceil n(1-\gamma) + \gamma f + 1 \rceil = 8$ lists, and \emph{solid} if it appears in $n-2f=11$ lists. A directed edge $(u \to v)$ also requires at least $\theta_{\text{nb}}$ votes.

\textbf{Workload and Metrics.}
We evaluate two key protocol knobs that control batch formation: $lo-size$, the maximum number of transactions per $LocalOrder$, and $lo-interval$ ($\tau$), the cadence in milliseconds at which replicas emit $LocalOrders$. We report two primary metrics: throughput, measured in committed TPS, and end-to-end latency, measured from the timestamp of a transaction's submission by a client to its final BFT commitment. Each data point reported is the median of 10 runs to ensure robustness.

\subsection{Performance under Varying Load and Batch Size}

\textbf{Experiment Design.}
We fix the local order interval at $\tau=250$\,ms and vary $lo-size$ within $\{50, 100, 200, 400\}$. We test under high ($tx-rate$=800), moderate ($tx-rate$=400), and low ($tx-rate$=200) client submission rates.

\textbf{Results.} The results are shown in Fig.~\ref{fig:comparison}. Across all load conditions, AUTIG consistently achieves higher throughput than Themis, while maintaining comparable or lower latency. 

Under high load, this advantage is most pronounced. AUTIG's throughput exhibits a clear "knee" around a $lo-size$ of 100, reaching a peak significantly higher than Themis's best performance. Most critically, at a small $lo-size$ of 50, Themis suffers from a pathological latency spike, with median latencies exceeding 10 seconds. AUTIG, in the same scenario, maintains a low and stable latency, demonstrating its robustness. Under moderate and low loads, AUTIG maintains a consistent and clear throughput advantage across all tested batch sizes. 

\textbf{Mechanistic Explanation.}
The observed trends are a direct consequence of how each protocol handles inter-replica transaction visibility under the constraints of its voting thresholds. In Themis, each proposal is computed de novo from only that round's $LocalOrders$. Under high load, transactions from many clients compete for inclusion. When $lo-size$ is small, each replica's list captures only a narrow, somewhat random slice of this transaction pool. This leads to poor cross-replica co-visibility within a single round's collection. Consequently, very few transactions meet the $\theta_{\text{nb}}=8$ threshold to become non-blank, and even fewer pairs accumulate enough directional votes to form edges. The resulting dependency graph is sparse and far from a tournament, preventing immediate linearization. Finalization is deferred, forcing the system to wait for subsequent rounds to provide "update" messages to orient missing edges. This multi-round deferral, which we term the blocks-to-finalize ($BTF$) factor, becomes very large ($BTF \gg 1$), causing the extreme latency spike. As $lo-size$ increases, the co-visibility improves, more transactions cross the thresholds within a single round, $BTF$ shrinks toward 1, and performance rapidly improves until a "knee" is reached. Beyond this point, further increases in $lo-size$ mainly add redundant information and increase the $O($lo-size$^2)$ accounting overhead without enlarging the finalizable prefix.

AUTIG's stateful UTIG breaks this dependency on single-round co-visibility. It maintains a persistent frontier of all unconfirmed transactions and accumulates weights across rounds. Even if a pair $\{u,v\}$ receives only a few votes in the current round, this evidence is added to its historical total. This cumulative process ensures that the graph of active transactions is almost always dense and ready for extraction, keeping $BTF_{\text{utig}} \approx 1$ across a wide range of parameters. This is why AUTIG avoids the pathological latency tail and can extract a larger, stable-sized prefix each round.

\subsection{Impact of Local Order Interval}

\textbf{Experiment Design.}
We investigate the systems' sensitivity to the emission frequency $\tau$, fixing the workload at $tx-rate$=400 tx/s and $lo-size$=100, while varying $lo-interval$ from 25\,ms to 350\,ms.

\textbf{Results.}
The results are shown in Fig.~\ref{fig:latency_interval}. For AUTIG, latency decreases roughly linearly as $\tau$ shrinks. For Themis, throughput falls as $\tau$ decreases and latency rises once $\tau<100$\,ms; at $\tau=25$\,ms we observe a very high latency. At that point, AUTIG delivers several$\times$ higher throughput and over an order-of-magnitude lower latency.

\textbf{Mechanistic Explanation.}
We provide a compact analytical model in Appendix~\ref{app:lo-interval}. Briefly, for Themis small $\tau$ reduces per-round co-visibility, causing $BTF(\tau)$ to blow up; for AUTIG, the persistent UTIG keeps $BTF_{\text{utig}}\!\approx\!1$, making latency nearly linear in $\tau$.

\subsection{Summary of Experimental Findings}
By replacing Themis's stateless, per-round graph reconstruction with a stateful, incremental UTIG, AUTIG effectively solves the "co-visibility bottleneck." It accumulates ordering evidence across rounds, ensuring that a large, finalizable prefix of transactions is almost always available for extraction. This leads to higher throughput, lower and more predictable latency, and greater robustness to variations in workload and protocol parameters.

\begin{figure}[t]
  \centering
  \begin{subfigure}{0.8\linewidth}
    \centering
    \includegraphics[width=\linewidth]{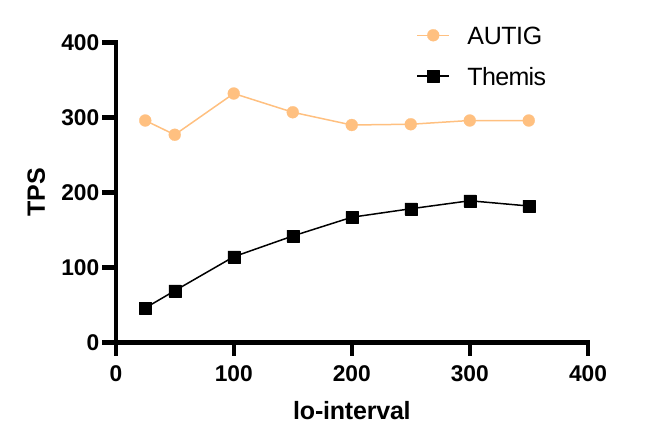}
    \caption{TPS ($tx-rate$=400, $lo-size$=100)}
    \label{fig:themis}
  \end{subfigure}

  \vspace{0.4em} 

  \begin{subfigure}{0.8\linewidth}
    \centering
    \includegraphics[width=\linewidth]{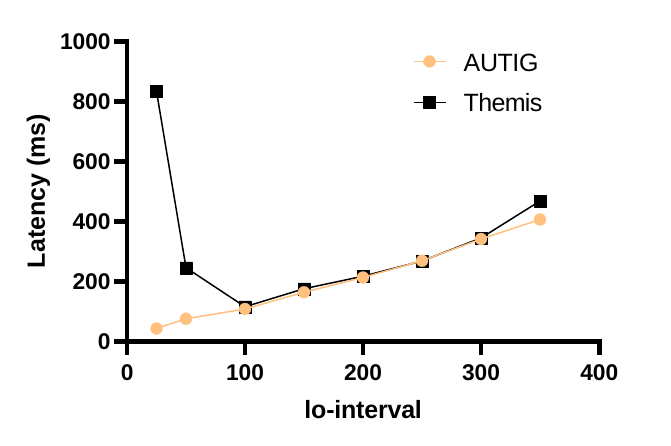}
    \caption{Latency ($tx-rate$=400, $lo-size$=100)}
    \label{fig:autig}
  \end{subfigure}

  \caption{TPS and latency of AUTIG and Themis under varying $lo-interval$.}
  \label{fig:latency_interval}
\end{figure}

\section{Conclusion}
We introduced AUTIG, a proof-carrying fair-ordering service that replaces symmetric re-execution with an asymmetric design: a leader maintains a persistent UTIG for incremental ordering while followers perform stateless proof audits. This removes system-wide, per-round graph reconstruction. AUTIG typically improves throughput and stabilizes latency while preserving \(\gamma\)-batch-order fairness. Limitations include proof size growth with the finalized prefix and sensitivity to pruning/handoffs; promising directions include commitment/zk-based proof compression, adaptive pruning, dynamic membership, and cross-shard ordering. AUTIG illustrates a design principle for BFT services: concentrate heavy state and computation where they can be amortized, and make the result efficiently auditable elsewhere; this asymmetry can help overcome the scalability barriers of symmetric verification.



\cleardoublepage
\appendix
\section*{Ethical Considerations}
\textbf{Stakeholders and benefits.}
Our work targets BFT ordering services used by permissioned/permissionless ledgers. Primary stakeholders are end users, application developers, validators/replicas, and exchange operators. By constraining leader discretion and reducing redundant re-execution, AUTIG aims to lower MEV-style value extraction and improve fairness and latency for all users.

\textbf{Risks and potential negative impacts.}
(1) \emph{Adversarial workloads.} Releasing ordering code could help adversaries craft worst-case inputs (e.g., SCC chains) to stress the system. We mitigate with liveness anchors, pruning, and formal verification conditions; we also include stress-test harnesses so operators can reproduce and harden against such patterns. 
(2) \emph{Security misconfiguration.} Incorrect deployment (thresholds, keys, or leader-handoff policy) could degrade guarantees. Our artifact includes default-safe configs, invariants, and automated checks.

\textbf{Data, privacy, and human subjects.}
We used synthetic workloads only; no personal data or user identifiers are processed. No human-subjects experiments were conducted; thus no IRB approval was required.

\textbf{Dual use and disclosure.}
The artifact ships with adversarial scenarios and detectors so operators can reproduce and patch issues. We follow responsible disclosure if any vulnerability in third-party stacks is encountered during reproduction.

\textbf{Environmental considerations.}
Experiments report machine counts and instance classes to make resource use transparent; the artifact provides small-scale repro scripts to limit compute/energy where full-scale replication is unnecessary.

\section* {Secure UTIG Pruning}
\label{app:pruning}
To remain efficient over long-term operation, AUTIG incorporates a secure pruning mechanism. The design is guided by a key principle: pruning is a leader-side optimization that must be invisible and inconsequential to the followers' stateless verification process. This is achieved by ensuring that any pruned information is provably irrelevant to the extraction of the current round's finalizable order prefix. The safety of this mechanism is guaranteed by three fundamental invariants. \textbf{Reproducibility Invariant:} Any proposed final order $F$ must be fully derivable from the information contained within its corresponding verifiable fragment $\mathcal{F}$. This ensures that verification is a self-contained process, completely independent of any state that the leader might have pruned locally. \textbf{Live-Core Invariant:} Before extracting an order in any round, the leader's UTIG must contain all nodes and edges necessary to correctly compute the finalizable prefix up to the solid anchor. Pruning must not remove information that is critical for the current round's liveness decision. \textbf{Monotone-Weight Invariant:} The cumulative weights $W(\cdot,\cdot)$ are append-only; they are never decreased or reset. This preserves all historical evidence for pairwise orderings, even if a transaction is temporarily deactivated from the active graph.

Governed by these invariants, the leader applies a set of admissible pruning rules that align with the deferred-ordering logic of prior art. First, any transaction included in a committed final order $F$ is immediately removed. Second, nodes that are classified as \(\Blank\) in the current batch can be pruned, as they have no incident edges by definition and thus do not affect the graph's connectivity. Third, after identifying the solid anchor, any \(\Shaded\) nodes that lie strictly after the anchor in the topological order may be pruned, optionally gated by a staleness horizon of $H$ rounds. Here, \(H\) acts as a safety-neutral hysteresis knob: a node becomes eligible only after remaining \(\Shaded\) and post-anchor for \(H\) consecutive extractions; larger \(H\) merely delays pruning to absorb late evidence. With soft pruning, fairness and liveness are unchanged because weights are append-only and pruned nodes can be reactivated.

By default, AUTIG implements soft pruning: admissible vertices and edges are removed from the active working set, but their historical weights and last-seen metadata are retained. This allows a pruned transaction to be reactivated in $O(1)$ if it reappears in a future batch, thus preserving liveness. Crucially, this entire pruning process is invisible to followers. Since verification relies solely on the self-contained proof, a faulty leader cannot exploit pruning to bias the final order; any such attempt would result in a proof that fails the completeness and anchor checks during verification.

\section*{Leader Handoff and State Recovery}
\label{app:handoff}
\subsubsection*{Order Fault Detection and Handoff Trigger}
If a follower replica's stateless verification (Algorithm~\ref{alg:verification}) of a received fragment $\mathcal{F}$ fails, it has deterministic proof of the leader's misconduct. It then constructs and broadcasts a signed $OrderFault$ message, containing the digest $H_{\mathcal{F}}$ and the original leader's signature on it as evidence. Because verification is deterministic and based on public rules, a correct Order Leader can never produce a fragment that honest replicas would deem invalid. Therefore, upon collecting $f+1$ such signed $OrderFault$ messages for the same digest $H_{\mathcal{F}}$, all correct replicas deterministically trigger the Order Leader Handoff Protocol and elect a successor (e.g., via a round-robin schedule). This threshold guarantees that at least one correct replica has detected the fault, ruling out spurious accusations from a minority of malicious nodes.

\subsubsection*{Lightweight State Recovery Protocol}
Upon election, the new Order Leader must reconstruct a valid UTIG to resume its duties. It executes a secure and efficient lightweight state recovery protocol. This protocol crucially avoids a direct state transfer from the old leader. Instead, it leverages the immutable history committed by the underlying BFT layer as a foundation for safe state reconstruction. The protocol proceeds in three steps. 

\textbf{Identify and Retrieve the Anchor Fragment:} The new Order Leader first inspects the chain of committed BFT blocks to find the last block containing a digest $H_{\mathcal{F}_{anchor}}$ of a valid order fragment. This digest serves as the universally agreed-upon anchor point. A fundamental property of BFT protocols is that replicas only vote for a block after receiving the full data corresponding to the digests within it. Therefore, the new leader can securely retrieve the complete anchor fragment, $\mathcal{F}_{anchor} = \langle F_{anchor}, \mathcal{L}_{\text{batch,anchor}}, \mathcal{P}_{anchor} \rangle$, from any correct replica by presenting the committed digest $H_{\mathcal{F}_{anchor}}$. The integrity of the retrieved $\mathcal{F}_{anchor}$ is verified by re-computing its digest and matching it against the one committed in the BFT chain. 

\textbf{Initialize UTIG from the Anchor Proof:} With the trusted anchor fragment $\mathcal{F}_{anchor}$ in hand, the new leader initializes a fresh UTIG. The recovery of state is a deterministic process derived from the proof $\mathcal{P}_{anchor}$. The new leader populates its UTIG as follows: Node and State Initialization. For every transaction $tx \in F_{anchor}$, it adds a node to its UTIG and sets its state according to the state assertion in $\mathcal{P}_{anchor}$. Historical Weight Reconstruction. For every edge assertion $(u,v, w_{uv}, w_{vu})$ in $\mathcal{P}_{anchor}$, the leader reconstructs the \textit{implied historical weight}. It first calculates the weight contribution from the batch, $W_{\text{batch}}(u,v)$, using $\mathcal{L}_{\text{batch,anchor}}$. Then, it sets the historical weight in its new UTIG as $G_{\text{utig}}.W(u,v) \leftarrow w_{uv} - W_{\text{batch}}(u,v)$. This value represents the minimum, verifiable historical support for the ordering $u \rightarrow v$ prior to the anchor batch. Edge Initialization. Based on these reconstructed weights and states, it initializes the directed edges $E$ within this restored portion of the graph. This process deterministically restores a partial but critical UTIG state that correctly reflects all ordering decisions up to the anchor point, ensuring a seamless transition of fairness guarantees. 

\textbf{Rapid State Convergence via Recovery Batch:}  The UTIG is now correctly initialized for all recently finalized transactions. To account for all other unconfirmed transactions, the new leader broadcasts a $RecoveryRequest$. This prompts all correct replicas to respond with a signed $LocalOrder$ containing all non-finalized transactions in their current view. By aggregating these responses into a single, large recovery batch and performing a large-scale update on its partially initialized UTIG, the new leader rapidly converges its graph to the collective state of the network.

For any pair \((u,v)\) present in the anchor fragment, the reconstructed historical weight
\(\widehat{W}^{\text{hist}}(u,v)=w_{uv}^{\text{anchor}}-W_{\batch}^{\text{anchor}}(u,v)\)
is the \emph{exact} cumulative weight before the anchor batch; subsequent rounds add to it append-only. 
Transactions not in the anchor start with \(W(\cdot,\cdot)=0\) and are (re)introduced by the recovery batch.
Frontier pairs in the anchor are used only to guarantee down-closure of the committed prefix; they do not create hidden state: verification in future rounds recomputes \(B_r\) from each round's \(\mathcal{L}_{\batch}\).

\section*{Proofs for §4: Correctness of Optimizations}
\label{app:optim_proofs}

\begin{lemma}[Correctness of Incremental Edge Updates]
\label{lem:incremental_correctness_app}
Let $W_{\text{new}}=W_{\text{old}}+W_{\batch}$ be cumulative weights after the round, let $DirtyNodes$ be transactions whose state toggles to/from $\Blank$, and define
\[
\begin{aligned}
DirtyPairs \;=\; \bigl\{\{u,v\}\,:\,
& W_{\batch}(u,v)+W_{\batch}(v,u)>0,\;\\
& \max\{W_{\text{new}}(u,v),W_{\text{new}}(v,u)\}\ge T_{\edge}\bigr\}.
\end{aligned}
\]
If Algorithm~\ref{alg:incremental_update} re-evaluates $P(\cdot)$ exactly on
\[
Affected \;=\; DirtyPairs \;\cup\; \bigl\{\{u,v\}\,:\, u\in DirtyNodes,\ v\in V_{\text{nb}},\ u\neq v\bigr\},
\]
then the resulting graph $\widetilde G_r$ is identical to the from-scratch graph $G_r^{\star}$ after round $r$.
\end{lemma}

\begin{proof}
Vertices and states are recomputed from $\mathcal L_{\batch}$ and match; weights are append-only and match by construction. An orientation $P(u\!\to\!v)$ can change only if a state flips (covered by $DirtyNodes$ adjacency) or if the pair appears in this batch and at least one direction lies in the threshold band where $P$ is decided; this is precisely $DirtyPairs$, which also covers orientation flips when both directions are $\ge T_{\edge}$. Pairs outside $Affected$ keep the same $P(\cdot)$ value. Hence re-orienting only $Affected$ yields the $E$ of $G_r^{\star}$.
\end{proof}

\begin{theorem}[Incremental Equivalence and Pipeline Linearizability; formal version of Thm.~\ref{thm:incremental_equiv}]
With notation as in §4, $\widetilde G_r\equiv G_r^{\star}$. Under a SWMR discipline on UTIG, the concurrent pipeline is linearizable to a per-round sequential execution with the serialization point at the end of Algorithm~\ref{alg:incremental_update}.
\end{theorem}

\begin{proof}
Equivalence follows from Lemma~\ref{lem:incremental_correctness_app}. For linearizability, let the updater acquire the write lock, apply all per-round deltas to $(V, \text{states}, W, E)$ atomically, and release. Proposer threads acquire a read lock and operate on an immutable snapshot; no writer is active concurrently. Define the history’s linearization by placing each proposal read right after the corresponding updater’s release. Mutual exclusion prevents overlapping writes; readers see a prefix-closed snapshot consistent with that point. Hence the concurrent history is observationally equivalent to a sequential execution that, in each round, first applies Algorithm~\ref{alg:incremental_update} then reads the graph to extract the prefix.
\end{proof}

\section*{Proofs for §\ref{sec:analysis}: Security and Liveness}
\label{app:sec6proofs}

\begin{proof}[Proof of Lemma~\ref{lem:graph_equiv}]
States for $y\in F$ are re-derived from $\mathcal L_{\batch}$ by the public map and must match; for every unordered $\{u,v\}\subseteq F$, the asserted totals dominate the recomputed batch counts, so applying $P(\cdot)$ to $(w_{uv},w_{vu})$ reproduces the leader’s orientation on $F$. The frontier enumerates the full Cartesian product $(B_r\setminus F)\times F$ and enforces $\neg P(x\!\to\!y)$ for all pairs, hence no active non-blank outside $F$ points into $F$, making $F$ a valid topological prefix.
\end{proof}

\begin{proof}[Proof of Theorem~\ref{thm:liveness_bounds}]
After GST, at most $\Delta$ to reach $n{-}2f$ correct replicas; the collection window adds $T_{\collect}{+}\Delta$, so solidity occurs by the next extraction. With an honest order leader, the solid-anchor rule finalizes a non-empty prefix per round; SCC resolution costs at most $R_{\defer}$ extra rounds. Byzantine order leaders are detected (Theorem~\ref{thm:verif_sc}) and replaced; with round-robin at most $f$ occur consecutively, yielding the bound.
\end{proof}

\begin{proof}[Proof of Theorem~\ref{thm:pruning_safety}]
Extraction uses only the non-blank induced subgraph; removing current \Blank\ vertices deletes no incident edges. Pruning $\Shaded$ strictly after the anchor cannot create a path entering the finalized prefix, so the cut is unchanged. Verification uses only $(F,\mathcal L_{\batch},\mathcal P)$; leader-side pruning is invisible. Soft-pruning preserves append-only weights, so reappearance reactivates the vertex without fairness regression.
\end{proof}

\begin{proof}[Proof of Theorem~\ref{thm:recovery}]
The committed digest authenticates the anchor fragment; recomputing $W_{\batch}^{\text{anchor}}$ from $\mathcal L_{\batch,\text{anchor}}$ and subtracting yields exact pre-anchor totals. A recovery batch with any $\ge n{-}f$ admissible replies includes $\ge n{-}2f$ honest inputs once; cumulative weights are append-only; malformed/duplicate/cross-round replies are rejected. Verification remains fragment-local by Theorem~\ref{thm:verif_sc}.
\end{proof}

\section*{Complexity }
\label{app:complexity}
Let $m$ be the per-$LocalOrder$ cap on listed transactions and $q\!\in\![n{-}2f,n{-}f]$ the number of admissible orders per round; then $|\{(u,v):W_{\batch}(u,v)>0\}|\le q\,m(m{-}1)/2=\mathcal O(m^2 n)$. The leader’s amortized per-round work is $\tilde{\mathcal O}\!\bigl(m^2 n + |\mathrm{Adj}(\text{DirtyNodes})| + |E_{\text{ext}}|\bigr)$ (pairs touched by the batch, adjacencies of state-flipping nodes, plus linear-time SCC/condensation/topo on the extraction subgraph). The follower’s per-fragment verification cost is $\tilde{\mathcal O}\!\bigl(|F|^2 + |F|\,(|B_r|-|F|)\bigr)$ plus batch re-count time, linear in the proof size.

\section*{Mechanistic Explanation}
\label{app:lo-interval}
This divergence again stems from the core architectural designs, which can be modeled by analyzing the trade-offs. Let $\tau$ be the $lo-interval$.

For Themis, performance is governed by the interplay between proposal cadence and per-round information density. A transaction's latency and the system's throughput can be approximated by \(\mathrm{LAT}_{\mathrm{themis}}(\tau)\approx\tfrac{\tau}{2}+BTF(\tau)\cdot\tau+\varepsilon\) and \(\mathrm{TPS}_{\mathrm{themis}}(\tau)\approx\lambda_{\mathrm{prop}}(\tau)\cdot\tfrac{M(\tau)}{BTF(\tau)}\). Here, $\tau/2$ is the average sampling delay, $BTF(\tau)$ is the blocks-to-finalize factor, $M(\tau)$ is the finalizable mass per proposal, and $\lambda_{\mathrm{prop}}(\tau) \lesssim 1/\tau$ is the effective proposal rate. When $\tau$ is very small, each replica's observation window is thin, causing poor co-visibility. This makes the per-round graph sparse, drastically increasing $BTF(\tau)$ and shrinking $M(\tau)$. The explosive growth in the $BTF(\tau) \cdot \tau$ term dominates latency, and the collapse of the $M(\tau)/BTF(\tau)$ ratio crushes throughput, explaining the poor performance at small $\tau$. TPS improves as $\tau$ grows because the improvement in per-round graph density (shrinking $BTF$, growing $M$) outweighs the slower proposal cadence, until the graph is dense enough and performance saturates.

For AUTIG, the stateful UTIG keeps the graph persistently dense, ensuring \(BTF_{\mathrm{utig}}\approx 1\). Its performance envelope is thus \(\mathrm{LAT}_{\mathrm{utig}}(\tau)\approx\tfrac{\tau}{2}+c\,\tau\) (with small \(c\)) and \(\mathrm{TPS}_{\mathrm{utig}}(\tau)\approx\lambda_{\mathrm{prop}}(\tau)\cdot M_{\mathrm{utig}}(\tau)\). Latency becomes a near-linear function of $\tau$, dominated by the sampling delay, which matches our observations perfectly. Throughput becomes a product of two countervailing forces: decreasing $\tau$ raises the proposal cadence ($\lambda_{\mathrm{prop}}$), but also shortens the time for transactions to accumulate support, potentially shrinking the per-proposal finalizable mass ($M_{\mathrm{utig}}$) and amplifying fixed overheads. These two effects largely balance out, resulting in a throughput curve that is relatively flat and only weakly sensitive to $\tau$. This robustness to high-frequency reporting is a key advantage of the stateful design.

\cleardoublepage

\section*{Open Science}

\textbf{Artifacts (anonymized).}
We provide: (i) AUTIG and Themis source code implementations with complete consensus protocols;
(ii) configurations for both n=5, f=1; 
(iii) a built-in workload generator (client) with fixed seeds and tunable parameters (--tx-rate, --lo-size, --lo-interval); 
(iv) comprehensive documentation including README.md with 5-node EC2 deployment guide, expected outputs, and troubleshooting instructions;
(v) repository structure with Themis and AUTIG folders.

\textbf{Access.}
An anonymized archive is available for the entire review period at:
\url{https://anonymous.4open.science/r/fair-ordering-FB81/README.md}.

\textbf{Requirements.}
AWS EC2 virtual machines with Go 1.22+ installed; no containers or additional dependencies. 
Security groups must allow inter-node TCP communication on port 8000. 
Both protocols use identical configuration parameters and deployment procedures.

\cleardoublepage
\bibliographystyle{plain}
\bibliography{refs}

\end{document}